\documentclass[twoside,a4paper,12pt]{article}
\usepackage{amsmath,amssymb,dsfont,vmargin,amsthm,cases,stmaryrd}  
\usepackage[english]{babel}
\usepackage[T1]{fontenc}    
\usepackage{graphicx}       
\usepackage{psfrag}         
\newtheorem{proposition}{Proposition}

\newtheorem{definition}{Definition}
\title{Conditional inference with a complex sampling: exact computations and Monte Carlo estimations}
\author{
        {\large François Coquet} and {\large \'{E}ric Lesage}\\
        { \small \textit{CREST(ENSAI) and IRMAR, Campus de Ker Lann, F-35172 BRUZ}}\\
            \small \textit{Université européenne de Bretagne, France}
        }
\begin{document}
\maketitle
\begin{abstract}
In survey statistics, the usual technique
for estimating a population
total consists in summing appropriately weighted variable values for
the units in the sample. Different weighting systems exit: sampling
weights, GREG weights or calibration weights for example.
\\
In this article, we propose to use the
inverse of conditional
inclusion probabilities as weighting system. We study examples where
an auxiliary information enables to perform an a posteriori
stratification of the population. We show that, in these cases,
exact computations of the conditional weights are possible.
\\
When the auxiliary information consists in the knowledge of a
quantitative variable for all the units of the population, then we
show that the conditional weights can be estimated via Monte-Carlo
simulations. This method is applied to outlier and strata-Jumper
adjustments.
\end{abstract}


\section{Introduction}
The purpose of this article is to give a
systematic use of the auxiliary information at the
estimation phase by the means of Monte Carlo methods, in a design based approach. \\

\medskip
In survey sampling, we often face a situation where we use
information about the population (auxiliary information) available
only at the estimation phase. For example, this information can be
provided by an administration file
 available only posterior to the collection stage.
Another example would be the number of respondents to a survey.
 It is classical to deal with the non-response mechanism by a second
 sampling phase (often Poisson sampling conditional to the size of the sample).
 The size of the respondents sample is known only after the collection.

 \medskip
This information can be compared to its counterpart estimated by the
means of the sample. A significant difference typically reveals an
unbalanced sample. In order to take this discrepancy into account,
it is necessary to re-evaluate our estimations. In practice, two
main technics exist: the model-assisted approach (ratio estimator,
post-stratification estimator, regression estimator) and the
calibration approach. The conditional approach we will develop in
this article has been so far mainly a theoretical concept because it
involves rather complex computations of the inclusion probabilities.
The use of Monte-Carlo methods could be a novelty that would enable
the use of conditional approach in practice.
In particular, it seems to be very helpful for the treatment of outliers and strata jumpers. \\

Conditional inference in survey sampling means that, at the
estimation phase, the sample selection is modelized by means of a
conditional probability. Hence, expectation and variance of the
estimators are computed according to this conditional sampling
probability. Moreover, we are thus provided with conditional
sampling weights with better properties than the original sampling
weights, in the sense
that they lead to a better balanced sample (or calibrated sample).\\

Conditional inference is not a new topic and several authors have
studied the conditional expectation and variance of estimators,
among them: Rao (1985), Robinson(1987), Tillé (1998, 1999) and
Andersson (2004). Moreover, one can see that the problematic of
conditional inference is close to inference in the context of
rejective sampling design. The difference is that in rejective
sampling, the conditioning event is controlled by the design,
whereas, in conditional
inference, the realization of the event is observed. \\

In section 2, the classical framework of finite population sampling and some notations are presented. \\

In section 3, we discuss the well-known setting of simple random
sampling where we condition on the sizes of the sub-samples on
strata (a posteriori stratification). This leads to an alternative
estimator to the classical HT estimator. While a large part of the
literature deals with the notion of correction of conditional bias,
we will directly use the concept of conditional HT estimator (Tillé,
1998), which seems more natural under conditional inference. A
simulation study will be performed in order to compare the accuracy
of the conditional strategy
to the traditional one.\\

In section 4, the sampling design is a Poisson sampling conditional to sample size $n$ (also called conditional Poisson
sampling of size $n$). We use again the information about the sub-samples sizes
to condition on. We show that the conditional probability corresponds exactly
to a stratified conditional Poisson sampling and we give recursive formula
that enables the calculation of the conditional inclusion probabilities. These results are new.\\

In section 5, we use a new conditioning statistic. Following Tillé
(1998, 1999), we use the non-conditional HT estimation of the mean
of the auxiliary variable to condition on. Whereas Tillé uses
asymptotical arguments in order to approximate the conditional
inclusion probabilities, we prefer to perform Monte Carlo
simulations to address a non-asymptotic setting. Note that this idea
of using independent replications of the sampling scheme in order to
estimate inclusion probabilities when the
sampling design is complex has been already proposed by Fattorini (2006) and Thompson and Wu (2008).\\

In section 6, we apply this method to practical examples: outlier and strata jumper in  business survey. This new method to deal with outliers gives good results.

\section{The context}

Let $U$ be a finite population of size $N$. The statistical units of
the population are indexed by a label $k \in \{1, ..., N\}$. A
random sample without replacement $s$ is selected using a
probability (sampling design) $p(.)$. $\mathcal{S}$ is the set of
the possible samples $s$. $I_{[k \in s]}$ is the indicator variable
which is equal to one when the unit $k$ is in the sample and $0$
otherwise. The size of the sample is $n(s)=|s|$. Let $B_k=\{s \in
\mathcal{S}, k \in s \}=\{s \in \mathcal{S}, I_{[k \in s]}=1 \}$ be
the set of samples that contain $k$. For a fixed individual $k$, let
$\pi_k=p(B_k)$ be the inclusion probability and let $d_k=\frac{1}{\pi_k}$ be its
sampling weight. For any variable $z$ that takes the value $z_k$ on the $U$-unit $k$,
the sum $t_z=\sum_{k \in U}z_k$ is referred to as the total of $z$ over $U$.
$\widehat{t}_{z, \pi}=\sum_{k \in s}\frac{1}{\pi_k}z_k$ is the Horvitz-Thompson estimator of the total $t_z$. \\
\newline
Let $x$ be an auxiliary variable that takes the value $x_k $ for the
individual $k$. The $x_k$ are assumed to be known for all the units
of $U$. Such auxiliary information is often used at the sampling
stage in order to improve the sampling design. For example, if the
auxiliary variable is a categorical variable then the sampling can
be stratified. If the auxiliary variable is quantitative, looking
for a balanced sampling on the total of $x$ is a natural idea. These
methods reduce the size of the initial set of admissible samples.
In the second example, $\mathcal{S}_{balanced}=\{s \in \mathcal{S}, \widehat{t}_{x, \pi}=t_x\}$. \\
We wish to use auxiliary information  \emph{after} the sample
selection, that is to take advantage of information such as the
number of units sampled in each stratum or the estimation of the
total $t_x$ given by the Horvitz-Thompson estimator. Let us take an
example where the sample consists in $20$ men and $80$ women, drawn
by a simple random sampling of size $n=100$ among a total population
of $N=200$ with equal inclusion probabilities $\pi_k=0.5$. And let
us assume that we are given \textit{a posteriori} the additional
information that the population has $100$ men and $100$ women. Then
it is hard to maintain anymore that the inclusion probability for
both men and women was actually $0.5$. It seems more sensible to
consider that the men sampled had indeed a inclusion probability of
$0.2$ and a weight of $5$. Conditional inference aims at giving some
theoretical
support to such feelings. \\

We use the notation $\mathbf{\Phi}(s)$ for the statistic
that will be
used in the conditioning. $\mathbf{\Phi}(s)$ is a random vector that
takes values in $\mathbb{R}^{q}$. In fact, $\mathbf{\Phi}(s)$ will
often be a discrete random vector which takes values in $\{1,...,
n\}^{q}$. At each possible subset $\mathbf{\varphi} \subset
\mathbf{\Phi}(\mathcal{S}) $  corresponds an event
$A_\varphi=\mathbf{\Phi}^{-1}(\mathbf{\varphi}) = \{s \in
\mathcal{S}, \mathbf{\Phi}(s) \in \mathbf{\varphi}\}$.

For example, if the auxiliary variable $x_k$ is the indicator
function of a domain, say $x_k=1$ if the unit $k$ is a man, then we
can choose $\mathbf{\Phi}(s)=\sum_{k \in s}I_{[k \in
domain]}=n_{domain}$ the sample size in the domain (number of men in
the sample). If the auxiliary variable $x_k$ is a quantitative
variable, then we can choose $\mathbf{\Phi}(s)=\sum_{k \in
s}\frac{x_k}{\pi_k}=\widehat{t}_{x, \pi}$ the
Horvitz-Thompson estimator of the total $t_x$.\\

\section{A posteriori Simple Random Sampling Stratification}
\subsection{Classical Inference}
In this section, the sampling design is a simple random sampling
without replacement(SRS) of fixed size $n$; $\mathcal{S}_{SRS}=\{s
\in \mathcal{S}, n(s)=n\}$; $p(s)= 1/\binom{N}{n}$ and the inclusion
probability of each individual $k$ is $\pi_k=n/N$. Let $y$ be the
variable of study. $y$ takes the value $y_k $ for the individual
$k$. The $y_k$ are observed for all the units of the sample. The
Horvitz-Thompson (HT) estimator of the total $t_y=\sum_{k \in U}y_k$
is
$\hat{t}_{y, HT}=\sum_{k \in U}\frac{y_k}{\pi_k}I_{[k \in s]}$. \\

Assume now that the population $U$ is split into $H$ sub-populations $U_h$ called strata.
Let $N_h=|U_h|$, $h \in \{1,..., H\}$ be the auxiliary information to be taken into account. We split the sample $s$ into
$H$ sub-samples $s_h$ defined by $s_h=s \cap U_h$. Let $n_h(s)=|s_h|$ be the size of the sub-sample $s_h$. \\
\newline
Ideally, to use the auxiliary information at the sampling stage
would be best. Here, a simple random stratified sampling (SRS
stratified) with a proportional allocation $N_hn/N$ would be more
efficient than a SRS.
For such a SRS stratified, the set of admissible samples is
$\mathcal{S}_{SRS stratified}=\{s \in \mathcal{S}, \forall h \in [1,
H], n_h(s)=N_hn/N \}$, and the sampling design is
$p(s)=\prod_{h \in [1,H]} \frac{1}{\binom{N_h}{n_h}},\quad s\in \mathcal{S}_{SRS stratified}$. Once again, our point is
precisely to consider setting where the auxiliary information becomes available \emph{posterior} to this sampling stage\\

\subsection{Conditional Inference}
The \emph{a posteriori} stratification with an initial SRS was
described by Rao(1985) and Tillé(1998). A sample $s_0$ of size
$n(s_0)=n$ is selected. We observe the sizes of the strata
sub-samples: $n_h(s_0)=\sum_{k \in U_{h}}I_{[k \in s]}$, $h \in [1,
H]$. We assume that $\forall h, n_h(s_0)>0$. We then consider the
event:
$$A_0 = \{ s \in \mathcal{S}, \forall h \in [ 1, H ], n_h(s)=n_h(s_0) \}. $$
It is clear that $s_0 \in A_0$, so $A_0$ is not empty. \\

We consider now the conditional probability: $p^{A_0}(.)=p(./A_0)$
which will be used as far inference is concerned. The conditional
inclusion probabilities are denoted
$$\pi_k^{A_0} = p^{A_0}\left( [I_{[k \in s]}=1] \right) =
\mathbb{E}^{A_0} \left(I_{[k \in s]}\right) =p\left( [I_{[k \in
s]}=1] \cap A_0 \right) / p(A_0).$$ Accordingly, we define the
conditional sampling weights: $d_k^{A_0}=\frac{1}{\pi_k^{A_0}}$.

\begin{proposition}
\begin{enumerate}
  \item The conditional probability $p^{A_0}$ is the
  law of a stratified simple random sampling with allocation $\left(n_1(s_0), ..., n_H(s_0)\right)$,
  \item For a unit $k$ of the strata $h$:
$\pi_k^{A_0}=\displaystyle\frac{n_h(s_0)}{N_h}$ and
$d_k^{A_0}=\displaystyle\frac{N_h}{n_h(s_0)}$.
\end{enumerate}
\end{proposition}

\begin{proof}
$|A_0| = \binom{N_1}{n_1(s_0)}\times...\times \binom{N_H}{n_H(s_0)} $. \\
$\forall s \in A_0$, $p^{A_0}(s)=1/|A_0|$. So we have:
\begin{eqnarray*}
p^{A_0}(s) & = & I_{[s \in A_0]} \frac{1}{\prod_{h \in [1,H]}\binom{N_h}{n_h(s_0)}}\\
& = &I_{[s \in A_0]} * \prod_{h \in [1,H]}\frac{1}{\binom{N_h}{n_h(s_0)}}\\
& = &\prod_{h \in [1,H]}I_{[n_h(s)=n_h(s_0)]} * \frac{1}{\binom{N_h}{n_h(s_0)}}
\end{eqnarray*}
and we recognize the probability law of a stratified simple random sampling
with allocation $\left(n_1(s_0), ..., n_H(s_0)\right)$.\\
2. follows immediately.
\end{proof}

Note that
$$\mathbb{E}^{A_0} \left(\sum_{k \in U}\frac{y_k}{\pi_k}I_{[k \in
s]}\right)= \sum_{k \in U}\frac{y_k}{\pi_k}\pi_k^{A_0} =\sum_h
\sum_{k \in U_h} y_k \frac{N n_h(s_0)}{n N_h}, $$
 so that the
genuine HT estimator is conditionally biased in this framework.

Even if, as Tillé(1998) mentioned, it is possible to correct this
bias simply by retrieving it from the HT estimator, it seems more
coherent to use another linear estimator constructed like the HT
estimator but, this time, using the conditional inclusion
probabilities.

\medskip
Remark that in practice $A_0$ should not be too small. The idea is that
for any unit $k$, we should be able to find a sample $s$ such that $s \in A_0$
and $k \in s$. Thus, all the units of $U$ have a positive conditional inclusion probability. \\

\begin{definition}
  The \textbf{conditional HT estimator }is defined as:
  $$
  \hat{t}_{y, CHT}=\sum_{k \in U}\frac{y_k}{\pi_k^{A_0}}I_{[k \in s]}
  $$
\end{definition}

The conditional Horvitz-Thompson (CHT) etimator is obviously conditionally
unbiased and, therefore, unconditionally unbiased.\\

This estimator is in fact the classical post-stratification
estimator obtained from a model-assisted approach (see Särndal et
Al.(1992) for example). However, conditional inference leads to a
different
derivation of the variance, which appears to be more reliable as we will see in next subsection.\\

\subsection{Simulations}\label{simul}
In this part, we will compare the punctual estimations of a total according to two strategies:
(SRS design + conditional (post-stratification) estimator) and (SRS design + HT estimator).\\
\begin{figure}[h]
\centering
\psfrag{a} {Comparison between $\hat{\mu}_{x, CHT}$ and $\hat{\mu}_{x, CH}$}
\includegraphics[width=7cm]{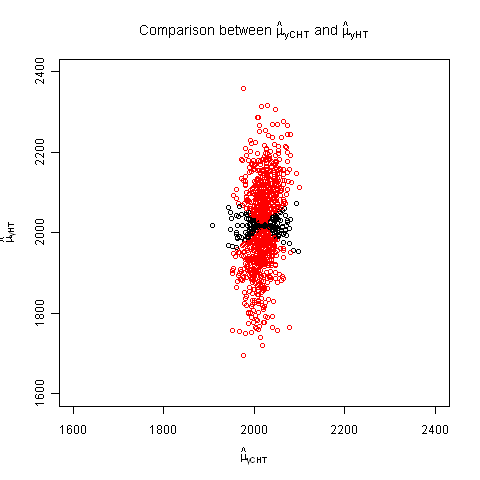}
\caption{Punctual Estimation}
\label{SRS, punctual}
\end{figure}
The population size is $N=500$, the variable $y$ is a quantitative
variable drawn from a uniform distribution over the interval $[0, 4
000]$. The population is divided into 4 strata corresponding to the
values of $y_k$ (if $y_k \in [0,1 000[$
then $k$ belongs to the strata 1 and so on ...). The auxiliary information will be the size of each strata in
the population. In this example, we get $N_1=123$, $N_2=123$, $N_3=132$ and $N_4=122$. \\

The finite population stays fixed and we simulate with the software
R $K=10^3$ simple random samples of size $n=100$. Two estimators of
the mean $\mu_y=\frac{1}{N}\sum_{k \in U}y_k$ are computed and
compared. The first one is the HT estimator: $\hat{\mu}_{y,
HT}=\frac{1}{n}\sum_{k \in s}y_k$ and the second one is the
conditional estimator:
$\hat{\mu}_{y, CHT}=\frac{1}{N}\sum_h \sum_{k \in U_h} y_k \frac{N_h }{n_h(s)}I_{[k \in s]}$.\\

On Figure \ref{SRS, punctual}, we can see the values of
$\hat{\mu}_{y, HT}$ and $\hat{\mu}_{y, CHT}$ for each of the $10^3$
simulations. The  red dots are those for which the conditional
estimation is closer to the true value $\mu_y=2019.01$ than the
unconditional estimation; red dots represents 83.5\% of the
simulations. Moreover, the empirical variance of the conditional
estimator
is clearly smaller than the empirical variance of the unconditional estimator.\\

This is completely coherent with the results obtained for the
post-stratification estimator in an model-assisted approach (see
Särndal et Al.(1992) for example). However, what is new and
fundamental in the conditional approach, is to understand that for
one fixed sample, the conditional bias and variance are much more
reliable than the unconditional bias and variance.
The theoretical study of the conditional variance estimation is a subject still to be developed.  \\

\subsection{Discussion}
\begin{enumerate}
  \item The traditional sampling strategy is defined as
  a couple (sampling design + estimator). We propose to define here the strategy as a triplet (sampling design + conditional sampling probability + estimator).
  \item We have conditioned on the event:
  $A_0 = \{ s \in \mathcal{S}, \forall h \in [1, H]~n_h(s)=n_h(s_0) \}$.
Under a SRS, it is similar to use the HT estimators of the sizes of
the strata in the conditioning, that is to use
$\mathbf{\Phi}(s)=(\hat{N}_1(s),..., \hat{N}_H(s))^t$, where
$\hat{N}_h(s)=\sum_{k \in U_h}\frac{I_{[k \in
s]}}{\pi_k}=\frac{N}{n}n_h(s)$. Then, $A_0 = \{ s \in \mathcal{S},
\mathbf{\Phi}(s)=\mathbf{\Phi}(s_0) \}$.  We will see in Section
\ref{phi} the importance of this remark.
  \item The CHT estimations of the sizes of the strata
  are equal to the true strata sizes $N_h$, which means that the CHT estimations,
  in this setting, have the calibration property for the auxiliary information of the size of the strata.
  Hence, conditional inference gives a theoretical framework for the current practice of calibration
  on auxiliary variables. \\
\end{enumerate}

\section{A Posteriori Conditional Poisson Stratification}

Rao(1985), Tillé(1999) and Andersson (2005) mentioned that a
posteriori stratification in a more complex setting than an an
initial SRS is not a trivial task, and that one must rely on
approximate procedures. In this section, we show that it is possible
to determine the conditional sampling design and to compute exactly
the conditional
inclusion probabilities for an a posteriori stratification with a conditional Poisson sampling of size $n$.\\

\subsection{Conditional Inference}

Let $\tilde{p}(s)= \prod_{k \in s}p_k    \prod_{k \in
\bar{s}}(1-p_k)$ be a Poisson sampling with inclusion probabilities
$\mathbf{p}=(p_1, \ldots, p_N)^{t}$, where $p_k \in ]0,1]$ and
$\bar{s}$ is the
complement of $s$ in $U$. Under a Poisson sampling, the units are selected independently. \\
By means of rejective technics, a conditional Poisson sampling of
size $n$ can be implemented from the Poisson sampling. Then, the
sampling design is:
$$p(s)=K^{-1}\mathbf{1}_{|s|=n} \prod_{k \in s}p_k    \prod_{k \in \bar{s}}(1-p_k), $$
where $K=\sum_{s, |s|=n}\prod_{k \in s}p_k    \prod_{k \in \bar{s}}(1-p_k)$.\\
The inclusion probabilities $\pi_k = f_k\left(U, \mathbf{p},
n\right)$ may be computed by means of a recursive method:
$$ f_k\left(U, \mathbf{p}, n\right)= \frac{p_k}{1-p_k}\frac{n}{\sum_{l \in U}\frac{p_l}{1-p_l}(1-f_l\left(U,
\mathbf{p}, n-1\right))}\left(1-f_k\left(U, \mathbf{p}, n-1\right)\right)$$
where $f_k\left(U, \mathbf{p}, 0\right)=0$.

This fact was proven by Chen et al.(1994) and one can also see
Deville (2000), Matei and Tillé (2005), and Bondesson(2010). An
alternative proof is given in Annex 1.

It is possible that the initial $\pi_k$ of the conditional Poisson sampling design are known
instead of the $p_k$'s. Chen et al.(1994) have shown that it is possible to inverse the functions
$f_k\left(U, \mathbf{p}, n\right)$ by the means of an algorithm which is an application
of the Newton method. One can see also Deville (2000) who gave an enhanced algorithm.\\
\newline
Assume that a posteriori, thanks to some auxiliary information, the
population is stratified in $H$ strata $U_h$, $h \in [1,H]$. The
size of the strata $U_h$ is known to be equal to $N_h$, and the size
of the sub-sample $s_h$ into $U_h$ is $n_h(s_0)>0$.
We consider the event $A_0 = \{ s \in \mathcal{S}, \forall h \in [ 1, H ], n_h(s)=n_h(s_0) \}$.\\
\medskip

\begin{proposition}
With an initial conditional Poisson sampling of size $n$:
\begin{enumerate}
  \item The probability conditional to the sub-samples sizes of the "a posteriori strata", $p^{A_0}(s)=p(s/A_0)$, is
   the probability law of a \textbf{stratified sampling }with
   (independent) \textbf{conditional Poisson sampling of size $n_h(s_o)$ in each stratum},
  \item The conditional inclusion probability $\pi^{A_0}_k$ of an element $k$
  of the strata $U_h$ is the inclusion probability of a conditional Poisson
  sampling of size $n_h(s_o)$ in a population of size $N_h$.
\end{enumerate}
\end{proposition}

\begin{proof}
1. For a conditional Poisson of fixe size n, a vector $(p_1, \ldots,
p_N)^{t}$ exists, where $p_k \in ]0,1]$, such that:
$$p(s)=K^{-1} \mathbf{1}_{|s|=n}\prod_{k \in s}p_k    \prod_{k \in \bar{s}}(1-p_k), $$
where $K=\sum_{s, |s|=n}\prod_{k \in s}p_k    \prod_{k \in \bar{s}}(1-p_k)$.\\
We remind that $A_0 = \{ s \in \mathcal{S}, \forall h \in [ 1, H ], n_h(s)=n_h(s_0) \}$\\
Then:
\begin{eqnarray*}
 p(A_0) & =& K^{-1} \tilde{p}\left(\bigcap_{h \in [ 1, H ]} [n_h(s)=n_h(s_0)]\right)\\
        &   = & K^{-1}\prod_{h \in [ 1, H ]}\tilde{p}\left( [n_h(s)=n_h(s_0)], \right)\\
\end{eqnarray*}
where, $\tilde{p}(.)$ is the law of the original Poisson sampling.
Let $s \in A_0$, then:
\begin{eqnarray*}
 p^{A_0}(s) &=&\frac{p(s)}{p(A_0)} \\
            &=& \frac{K^{-1} \prod_{h=1,\ldots, H}\left[\prod_{k \in s_h}p_k \prod_{k \in \bar{s}_h}(1-p_k)\right]}{K^{-1}\prod_{h \in [ 1, H ]}\tilde{p}( [n_h(s)=n_h(s_0)])}\\
            &=&\prod_{h=1,\ldots, H}
                        \frac{\prod_{k \in s_h}p_k \prod_{k \in \bar{s}_h}(1-p_k)}
                        {\tilde{p}( [n_h(s)=n_h(s_0)])}\\
            &=&\prod_{h=1,\ldots, H}
                        \frac{\prod_{k \in s_h}p_k \prod_{k \in \bar{s}_h}(1-p_k)}
                        {\sum_{s_h, |s_h|=n_h(s_0)}\prod_{k \in s-h}p_k    \prod_{k \in \bar{s}_h}(1-p_k)},
\end{eqnarray*}
which is the sampling design of a stratified
sampling with independent conditional Poisson sampling of size $n_h(s_o)$ in each stratum.\\
2. follows immediately.
\end{proof}

\begin{definition}
In the context of conditional inference on the
sub-sample sizes of
posteriori strata, under an initial conditional Poisson sampling of
size $n$, the \textbf{conditional HT estimator} of the total $t_y$
is: $$\hat{t}_{y,CHT}=\sum_{k \in s}\frac{y_k}{\pi^{A_0}_k}.$$
\end{definition}

The conditional variance can be estimated by means of one of the
approximated variance formulae
developed for the conditional Poisson sampling of size $n$. See for example Matei and Tillé(2005), or Andersson(2004).\\

\subsection{Simulations}

\begin{figure}[h]
\centering
\includegraphics[width=7cm]{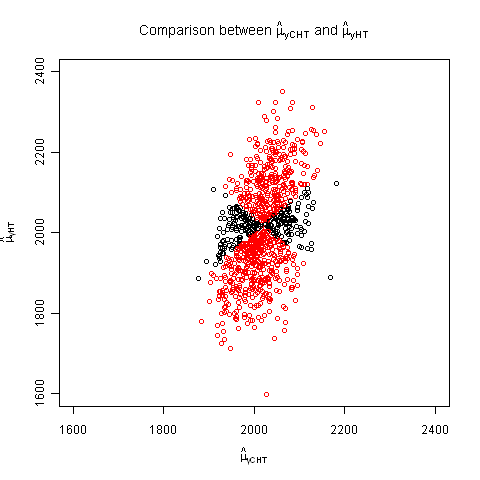}
\caption{Punctual Estimation}
\label{Poisson, punctual}
\end{figure}
 We take the same population as in subsection \ref{simul}.
 The sampling design is now a conditional Poisson sampling of size $n=100$.
 The probabilities $p_k$ of the underlying Poisson design have been generated
 randomly, in order that $\sum_{k \in U} p_k=n$ and $p_k \in [0.13 ; 0.27]$. \\
 $K=10^3$ simulations were performed. Figure \ref{Poisson, punctual} shows that
 the punctual estimation of the mean of $y$ is globally better for conditional inference.
 According to $77.3 \%$ of the simulations the conditional estimator is better
 than the unconditional estimator (red dots).
 The empirical variance as well is clearly better for the conditional estimator.

\subsection{Discussion}
 This method allows to compute exact conditional
 inclusion probabilities in an "a posteriori stratification" under conditional Poisson of size $n$.
 However, one can figure out that this method can be used for any unequal
 probabilities sampling design, had the sampling frame been randomly sorted.

\section{Conditioning on the Horwitz-Thompson estimator of an auxiliary variable}\label{phi}
In the previous sections, we used the sub-sample sizes in the strata
$n_h(s)$ to condition on. The good performances of this conditional
approach result from the fact that the sizes of the sub-sample are
important characteristics of the sample that are often used at the
sampling stage. So, it was not surprising that the use of this
information at the estimation stage would enhance
 the conditional estimators.\\

Another statistic that characterizes the representativeness of a
sample is its HT estimator of the mean $\mu_x$ (or total $t_x$) of
an auxiliary variable. This statistic is used at the sampling stage
in balanced sampling for example. So, as the sub-sample sizes into
the strata, this statistic should produce good results in a
conditional approach restraining the inference to the samples for
which the HT estimation
of $\mu_x$ are equal to the value $\hat{\mu}_0 = \hat{\mu}_{x, HT}(s_0)$ of the selected sample $s_0$ .\\

In fact, we want the (conditional) set of the possible samples to be
large enough in order that all conditional inclusion probabilities
be different from zero. It is therefore convenient to consider the
set of samples that give HT estimations not necessarily strictly
equal to $\hat{\mu}_0$ but close to $\hat{\mu}_0$.
Let $\varphi=[\hat{\mu}_0-\varepsilon, \hat{\mu}_0+\varepsilon]$, for some $\varepsilon >0$. \\

The set $A_\varphi$ of possible samples in our conditional approach
will be:
$$   A_\varphi=  \{ s \in \mathcal{S}, \: \hat{\mu}_{x, HT}(s) \in [\widehat{\mu}_0-\varepsilon, \hat{\mu}_0+\varepsilon]   \} .$$
The conditional inclusion probability of a unit $k$ is:
\begin{eqnarray*}
\pi_k^{A_\varphi} & = & p \left(  [k \in s] /\:  \left[\hat{\mu}_{x, HT}(s)  \in [\hat{\mu}_0-\varepsilon, \hat{\mu}_0+\varepsilon]  \right]            \right)  \\
& = &\frac {p\left( \{s \in \mathcal{S},\: k \in s \text{ and }
\hat{\mu}_{x, HT}(s)  \in [\hat{\mu}_0-\varepsilon,
\hat{\mu}_0+\varepsilon] \} \right)}          {p(A_\varphi)}.
\end{eqnarray*}
\\

If $\hat{\mu}_0=\mu_X$ then we are in a good configuration, because we are in
a balanced sampling situation and the $\pi_k^{A_\varphi}$ will certainly stay close to the $\pi_k$.\\
If $\widehat{\mu}_0 \gg \mu_X$ say, then the sample $s_0$ is
unbalanced, which means that in average, its units have a too large
contribution ${x_k}/{\pi_k}$, either because they are too big ($x_k$
large) or too heavy ($d_k=\frac{1}{\pi_k}$ too large). In this case,
the samples in $A_\varphi$ are also ill-balanced, because balanced
on $\widehat{\mu}_0$ instead of $\mu_X$: $\sum_{k \in s}
\frac{x_k}{\pi_k} \approx \widehat{\mu}_0 $. But conditioning on
this information will improve the estimation. Indeed, the
$\pi_k^{A_\varphi}$ will be different from the $\pi_k$. For example,
a unit $k$ with a big contribution ($\frac{x_k}{\pi_k}$ large) has
more chance to be in a sample of $A_\varphi$ than a unit $l$ with a
small contribution. So, we can expect that $\pi_k^{A_\varphi} >
\pi_k$ and $\pi_l^{A_\varphi} < \pi_l$.
And, in consequence, the
conditional weight $d_{k}^{\varphi}$ will be lower
than $d_{k}$ and $d_{l}^{\varphi}$ higher than $d_{l}$, which will "balance" the samples of $A_\varphi$.\\
\newline
\textbf{Discussion}:
\begin{itemize}
\item we can use different ways in order to define the subset $\varphi$. One way is to use
the distribution function of $\Phi(s)$, denoted $G(u)$ and to define
$\varphi$ as a symmetric interval:
$$\varphi=\left[G^{-1}(\max \{G(\Phi(s_0))-\frac{\alpha}{2}, 0\}),
G^{-1}(\min \{G(   \Phi(s_0))+\frac{\alpha}{2}, 1    )\}\right],$$ where $\alpha= 5\%$ for example. \\

Hence,
$$A_\varphi=\{s \in \mathcal{S}, \Phi(s) \in \left[G^{-1}(\max
\{G(\Phi(s_0))-\frac{\alpha}{2}, 0\}),
G^{-1}(\min \{G(   \Phi(s_0))+\frac{\alpha}{2}, 1    )\}\right]\},$$
and $p(A_\varphi) \leq \alpha$.\\

As the \emph{cdf} $G(u)$ is unknown in general, one has to replace
it by an estimated \emph{cdf} of $\Phi(s)$, denoted
$\hat{G}_{K}(u)$, computed by
means of simulations. \\
\end{itemize}

\section{Generalization: Conditional Inference Based on Monte Carlo
simulations.}\label{MC}

In this section, we consider a general initial sample design $p(s)$
with the inclusion probabilities $\pi_k$. We condition on the event
$A_\varphi=\mathbf{\Phi}^{-1}(\mathbf{\varphi}) = \{s \in
\mathcal{S}, \mathbf{\Phi}(s) \in \mathbf{\varphi}\}.$ For example,
we can use $\Phi(s)=\sum_{k \in s}\frac{\mathbf{x}_k}{\pi_k}$ the
unconditional HT estimator of $t_x$ and $\varphi=[\varphi_1,
\varphi_2]$ an interval that contains $\Phi(s_0)=\sum_{k \in
s_0}\frac{\mathbf{x}_k}{\pi_k}$, the HT estimation of $t_x$ with the
selected sample $s_0$. In other words, we will take into account the
information
that the HT estimator of the total of the auxiliary variable $x$ lies in some region $\varphi$.\\
\newline
The mathematical expression of $\pi_k^{A_\varphi}$ is
straightforward:
$$\pi_k^{A_\varphi}=p([k \in s]/A_\varphi)
=\frac{\sum_{s}p(s)\mathbf{1}_{s\in A_\varphi} \mathbf{1}_{[k \in
s]}}{p(A_\varphi)}.$$ But effective computation of the
$\pi_k^{A_\varphi}$'s may be not trivial if the distribution of
$\mathbf{\Phi}$ is complex. Tillé(1998) used an asymptotical
approach to solve this problem when $\Phi(s)=\sum_{k \in
s}\frac{x_k}{\pi_k}\mathbf{1}_{[k \in s]}$;
he has used normal approximations for the conditional and unconditional laws of $\Phi$. \\
In the previous sections, we have given examples where we were able
to compute the $\pi_k^{A_\varphi}$'s (and actually the
$p^{A_\varphi}(s)$'s) exactly. In this section, we give a general
Monte Carlo method to compute the $\pi_k^{A_\varphi}$.

\subsection{Monte Carlo}


We will use Monte Carlo simulations to estimate
$\mathbb{E}(\mathbf{1}_{A_\varphi}\mathbf{1}_{[k \in s]})$ and
$\mathbb{E}(\mathbf{1}_{A_\varphi})$. We repeat independently $K$
times the sample selection with the sampling design $p(s)$, thus
obtaining a set of samples $(s_1, \ldots, s_K)$. For each simulation
$i$, we compute $\mathbf{\Phi}(s_i)$ and $I_{A_\varphi}(s_i)$. Then
we compute $N+1$ statistics:
\begin{eqnarray*}
    M^{A_\varphi} &=&    \sum_{i=1}^{K}  \mathbf{1}_{A_\varphi}(s_i)   \\
   \forall k \in U, M_k^{\varphi}     &=&  \sum_{i=1}^{K}  \mathbf{1}_{A_\varphi}(s_i) \mathbf{1}_{[k \in s_i]}
\end{eqnarray*}

We obtain a consistent estimator of $\pi_k^{A_\varphi}$, as $K\rightarrow +\infty$:
\begin{equation}\label{bayes2}
    \hat{\pi}_k^{A_\varphi} =   \frac{M_k^{\varphi}/K}{ M^{A_\varphi}/K}=\frac{M_k^{\varphi}}{M^{A_\varphi}}
\end{equation}

\subsection{Point and variance estimations in conditional inference}
\begin{definition}
The Monte Carlo estimator of the total $t_y$ is the conditional
Horvitz-Thompson estimator of $t_y$ after replacing the conditional
inclusion probabilities by their Monte Carlo approximations:
$$
\hat{t}_{y,MC}=\sum_{k \in s_0}\frac{1}{\hat{\pi}_k^{A_\varphi}}y_k
$$
The Monte Carlo estimator of the variance of $\hat{t}_{y,MC}$ is:
$$
\widehat{\mathbb{V}}(\hat{t}_{y,MC})=\sum_{k,l \in s_0}\frac{1}{\hat{\pi}_{k,l}^{A_\varphi}}\frac{y_k}{\hat{\pi}_k^{A_\varphi}} \frac{y_l}{\hat{\pi}_l^{A_\varphi}}
(\hat{\pi}_{k,l}^{A_\varphi} - \hat{\pi}_k^{A_\varphi} \hat{\pi}_l^{A_\varphi}),
$$
where
$$\hat{\pi}_{k,l}^{A_\varphi}
=\frac{\sum_{i=1}^{K}  t_y\mathbf{1}_{A_\varphi}(s_i) \mathbf{1}_{[k
\in s_i]} \mathbf{1}_{[l \in s_i]}} { \sum_{i=1}^{K}
\mathbf{1}_{A_{\varphi}}(s_i) }.$$

\end{definition}
Fattorini(2006) established that $\hat{t}_{y,MC}$ is asymptotically unbiased as $M^{A_\varphi} \rightarrow \infty $,
and that its mean squared error converges to the variance of $\hat{t}_{y, HT}$.\\

Thompson and Wu (2008) studied the rate of convergence of the
estimators $\hat{ \pi }_k^{A_\varphi}$ and of the estimator
$\hat{t}_{y,MC}$ following Chebychev's inequality. Using normal
approximation instead of the Chebychev's inequality gives more
precise confidence intervals. We have thus a new confidence interval
for $\hat{\pi}_k^{A_\varphi}$:
$$p\left(|\hat{\pi}_k^{A_\varphi}-\pi_k^{A_\varphi}|<F^{-1}((1-\alpha)/2)\sqrt{\frac{1}{4M^{A_\varphi}}}\right)
\leq \alpha,$$ where $F$ is the distribution function of the normal
law $\mathcal{N}(0,1)$.

As for the relative bias, standard computation leads to:
\begin{eqnarray}
p\left(
\frac{|\hat{t}_{y, CHT}-\tilde{t}_{y, CHT}|}{\hat{t}_{y, CHT}} \leq \varepsilon
\right)
&\geq&
1-4 \times \sum_{k \in s}\left[1-F\left(\frac{\epsilon}{1+\epsilon}\sqrt{M^{A_\varphi}\pi_k^{A_\varphi}}\right) \right] \nonumber\\
&\geq&
1-4  n \frac{1}{\sqrt{2 \pi}}\frac{1+\epsilon} { \sqrt{M^{A_\varphi}\epsilon^2 \pi_{0}}} \cdot e^{-\left(\frac{M^{A_\varphi}\epsilon^2}{(1+\epsilon)^2} \pi_0\right)},\label{rbnormal}
\end{eqnarray}

where $\pi_0=\min \{\pi_k^{A_\varphi}, k \in U\}$. We used the inequality $1-F(u)\leq\frac{1}{\sqrt{2 \pi}} \frac{e^{-u^2}}{u}$ which is verified for large $u$. \\

The number $K$ of simulations is set so that
$\sum_{i=1}^{K}I_{A_\varphi}(s_i)$  reaches a pre-established
$M^{A_\varphi}$ value. Because of our conditional framework, $K$ is
a stochastic variable which follows a negative binomial distribution
and we have $E(K)=\displaystyle\frac{M^{A_\varphi}}{p(A_\varphi)} $.
For instance, if $p(A_\varphi)=0.05=5\%$, with $M^{A_\varphi}=10^6$, we expect $E(K)=2.10^7$ simulations.\\

\section{Conditional Inference Based on Monte Carlo Method in Order to Adjust for Outlier and Strata Jumper}

We will apply the above ideas to two examples close to situations
that can be found in establishments surveys: outlier and strata
jumper.

We consider an establishments survey, performed in year "\emph{n+1}",
and addressing year "\emph{n}". The auxiliary information $x$ which
is the turnover of the year "\emph{n}" is not known at the sampling
stage but is known at the estimation stage (this information may
come from, say, the fiscal administration).

\subsection{Outlier}

In this section, the auxiliary variable $x$ is simulated following a
gaussian law, more precisely $x_k\sim \mathfrak{N}(8 ~000, (2~
000)^{2})$ excepted for unit $k=1$ for which we assume that
$x_1=50~000$. The unit $k=1$ is an outlier. The variable of interest
$y$ is simulated by the linear model
$$y_k~=~1000~+~0.2~x_k~+~u_k,$$ where $u_k \sim \mathfrak{N}(0,
(500)^{2})$, $u_k$ is independent from $x_k$. The outcomes are $\mu_x= 8~531$ and
$\mu_y=2~695$.
\\

We assume that the sampling design of the establishments survey is a
SRS of size $n=20$ out of the population $U$ of size $N=100$ and
that the selected sample $s_0$ contains the unit $k=1$. For this
example, we have repeated the sample selection until the unit $1$
has been selected in $s_0$.
\\
We obtain $\Phi(s_0)=\hat{\mu}_{x, HT}(s_0)=9~970$, which is $17\%$ over the true value $\mu_x= 8~531$ and
  $\hat{\mu}_{y, HT}(s_0)=3~039$ (recall that the true value of $\mu_y$ is $2~695$).
\\
We set $\Phi$ and $\varphi$ as in section \ref{phi} and we use Monte
Carlo simulations in order to compute the conditional inclusion
probabilities
$\hat{\pi}_k^{A_\varphi}$. Each simulation is a selection of a sample following
a SRS of size $n=20$from the fixed population $U$. Recall that the value of $x_k$
will eventually be known for any unit $k\in U$. \\

Actually, we use two sets of simulations. The first set
is performed in order to estimate the \emph{cdf} of the statistic
$\Phi(s)=\hat{\mu}_{x, HT}(s)$ which will be used to condition on.
This estimated \emph{cdf} will enable us to construct the interval
$\varphi$.  More precisely, we choose the interval $\varphi=[9~793,
10~110]$ by the means of the estimated \emph{cdf} of
$\Phi(s)=\hat{\mu}_{y, HT}(s)$ and so that
 $p\left(\left[\hat{\mu}_{x, HT}(s) \in [9~793, \hat{\mu}_{x, HT}(s_0) ]\right]\right)=\frac{\alpha}{2}=2.5\%
 =p\left(\left[\hat{\mu}_{x, HT}(s) \in [\hat{\mu}_{x, HT}(s_0), 10~110 ]\right]\right)$.

 $A_\varphi$ is then the set of the possible samples in our conditional approach:
$$   A_\varphi=  \{ s \in \mathcal{S}, \: \hat{\mu}_{x, HT}(s)  \in [9~793, 10~110]   \} .$$

 Note that $p\left(\left[\hat{\mu}_{x, HT}(s) \in [9~793,
 10~110]\right]\right)=\alpha=5\%$.  $A_\varphi$ typically contains samples that over-estimate the mean of $x$.\\

The second set of Monte Carlo simulations consists in $K=10^{6}$
sample selections with a SRS of size $n=20$ performed in order to
estimate the conditional inclusion probabilities
$\hat{\pi}_k^{A_\varphi}$. $49~782$ ($4.98\%$) simulated samples
fall in $A_\varphi$, and among them, $49~767$ samples contain the
outlier, which correspond to the estimated conditional inclusion
probability of the outlier: $\hat{\pi}_1^{A_\varphi}=0.9997$. It
means that almost all the samples of $A_\varphi$ contain the outlier
that is mainly responsible for the over-estimation because of its
large value of the variable $x$!
\\
The weight of the unit $1$ has changed a lot, it has decreased from
$d_k=\frac{1}{0.2}=5$ to $\hat{d}^{A_\varphi}_k=1.0003$. The
conditional sampling weights of the other units of $s_0$ are more
comparable to
  their initial weights $d_k=5$ (see Figure \ref{Outlier, Conditional Inclusion Probabilities}).
  \\

 The conditional MC estimator
  $\hat{\mu}_{y,MC}(s)=\frac{1}{N}\sum_{k \in s}\frac{y_k}{\hat{\pi}_1^{A_\varphi}}$
  leads to a much better estimation of $\mu_y$: $\hat{\mu}_{y,MC}(s_0)=2~671$.\\

\begin{figure}[h]
\centering
\includegraphics[width=7cm]{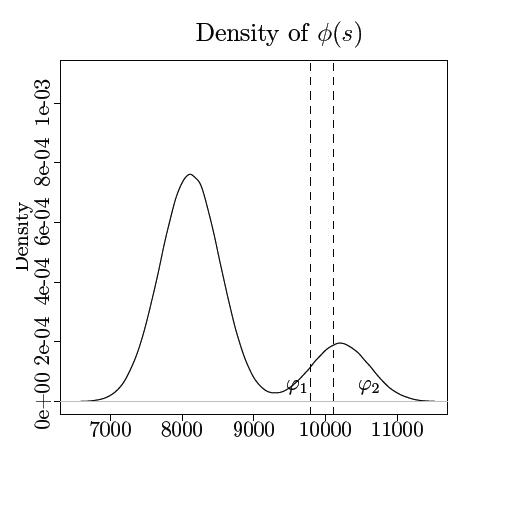}
\label{Outlier, Conditional Inclusion Probabilities}
\end{figure}

Figure \ref{Outlier, Conditional Inclusion Probabilities} gives an
idea of the conditional inclusion probabilities for all the units of
$U$. Moreover, this graph shows that the correction of the sampling
weights $\frac{\hat{d}^{A_\varphi}_k}{d_k}
=\frac{\pi_k}{\hat{\pi}_k^{A_\varphi}}$ is not a monotonic function
of $x_k$, which is in big contrast with calibration techniques which
only uses monotonic functions for weight correction purposes.\\

\begin{figure}[h]
\centering
\includegraphics[width=7cm]{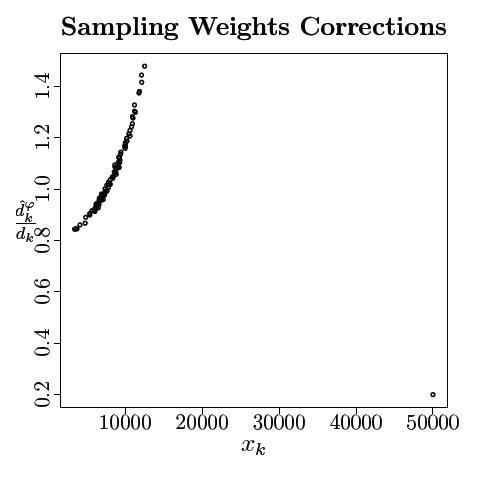}
\caption{Outlier, Density of $\Phi(s)=\hat{\mu}_{x, HT}(s)$}
\label{Outlier, density}
\end{figure}

A last remark concerns the distribution of the statistics
$\Phi(s)=\hat{\mu}_{x, HT}(s)$. Figure \ref{Outlier, density} shows
an unconditional distribution with 2 modes and far from gaussian.
This shows that in presence of outlier,
we can not use the method of Tillé (1999), which assumes a normal distribution for $\hat{\mu}_{x, HT}(s)$. \\

\subsection{Strata Jumper}

In this section, the population $U$ is divided into 2
sub-populations: the small firms and the large firms. Let us say
that the size is appreciated thanks to the turnover of the firm.
Official statistics have to be disseminated for this 2 different
sub-populations. Hence, the survey statistician has to split the
population into 2 strata corresponding to the sub-populations. This
may not be an easy job because the size of firms can evolve from one
year to another.

Here we assume that, at the time when the sample is selected, the
statistician does not know yet the auxiliary information $x$ of the
turnover of the firm for the year "\emph{n}", more precisely the
strata the firm belongs to for the year "\emph{n}". Let us assume
that he only knows this information for the previous
year,"\emph{n-1}". This information is denoted by $z$. In practice,
small firms are very numerous and the sampling rate for this strata is chosen low.
On the contrary, large firms are less numerous and their sampling rate is high. \\

When a unit is selected among the small firms but eventually happens
to be a large unit of year "\emph{n}", we call it a strata jumper.
At the estimation stage, when the information $x$ becomes available,
this unit will obviously be transferred to strata 2 . This will
bring a problem, not due to its $y$-value (which may well be typical
in strata 2) but to its sampling weight, computed according to
strata 1 (the small firms), and which will appear to
be very large in comparison to the other units in strata 2 at the estimation stage.\\

In our simulations, the population $U$ is split in 2 strata, by
means of the auxiliary variable $z$: $U^{z}_1$, of size
$N^{z}_1=10~000$,
is the strata of presumed small firms and $U^{z}_2$, of size $N^{z}_2=100$, the strata of presumed large firms. \\
The auxiliary variable $x$, which is the turnover of the year
"\emph{n}" known after collection, is simulated under a gaussian law
$\mathfrak{N}(8~000, (2~ 000)^{2})$ for the units of the strata
$U^{z}_2$ and for one selected unit of the strata $U^{z}_1$. Let us
say that this unit, the strata jumper, is unit $1$.

Our simulation gives $x_1=8~002$. The variable of interest $y$ is
simulated by the linear model $y_k~=~1000~+~0.2~x_k~+~u_k$, where
$u_k \sim \mathfrak{N}(0, (500)^{2})$, $u_k$ and $x_k$ independent.
We do not simulate the value of $x$ and $y$ for the other units of
the strata $U^{z}_1$ because we will focus on the estimation of the
mean of $y$ for the sub-population of large firms of year $n$
$U^{x}_2$: $\mu_{y,2}=\frac{1}{N_2}\sum_{k \in U^{x}_2} y_k$. We
find
$\mu_{x,2}=8~138$ and $\mu_{y,2}$ is $2~606$.\\

The sampling design of the establishments survey is a stratified SRS
of size $n_1=400$ in $U^{z}_1$ and $n_2=20$ in $U^{z}_2$. We assume
that the selected sample $s_0$ contains the unit $k=1$. In practice,
we repeat
the sample selection until the unit $1$ (the strata jumper) has been selected. \\

As previously, $\Phi$ and $\varphi$ are defined as in Section
\ref{phi}.

We use Monte Carlo simulations in order to compute
 the conditional inclusion probabilities $\hat{\pi}_k^{A_\varphi}$. A simulation is
 a selection of a sample with stratified SRS of size $n_1=400$ in $U^{z}_1$ and $n_2=20$ in $U^{z}_2$. \\

We choose the statistic $\Phi(s)=\hat{\mu}_{x,2, HT}(s)$ in order to
condition on. $K=10^{6}$ simulations are performed in
order to estimate the \emph{cdf} of $\Phi(s)$ and the conditional
inclusion probabilities.

Our simulations give $\Phi(s_0)=\hat{\mu}_{x,2, HT}(s_0)=9~510$,
which is far from the true value $\mu_{x,2}=8~138$
 and $\hat{\mu}_{y,2, HT}(s_0)=3~357$ (recall that the true value of $\mu_{y,2}$ is $2~606$).
 We choose the interval $\varphi=[8~961, 10~342]$ by the means of the estimated \emph{cdf} of $\Phi(s)=\hat{\mu}_{x,2, HT}(s)$ and so that $p\left(\left[\hat{\mu}_{x,2, HT}(s) \in [8~961, 10~342]\right]\right)=\alpha=5\%$. \\
$A_\varphi$ is then the set of the possible samples in our
conditional approach:
$$   A_\varphi=  \{ s \in \mathcal{S}, \: \hat{\mu}_{x, HT}(s)  \in [8~961, 10~342]   \} .$$
All samples in $A_\varphi$ over-estimate the mean of $x$.\\

Among the $10^6$ simulations, $49~778$ simulated samples ($4.98\%$)
belongs to $A_\varphi$. $55\%$ of them contains the strata jumper,
which gives the estimated conditional inclusion probability of the
strata jumper $\hat{\pi}_1^{A_\varphi}=0.55$. It is not a surprise
that the strata jumper is in one sample of $A_\varphi$ over two.
Indeed, its initial sampling
 weight $d_1=\frac{10~000}{400}=25 $ is high in comparison to the weights $d_k=\frac{100}{20}=5$
 of the other selected units
 of the strata $U^{x}_2$, and its contribution $\frac{25 x_1}{N_2}$ contributes to over-estimate the mean of $x$ .\\

 The conditional inclusion probabilities for the other units of $U^{x}_2$ are comparable to
 their initial $\pi_k=0.2$ (see Figure \ref{Strata Jumper, Conditional Inclusion Probabilities}).   \\

 The conditional MC estimator
  $\hat{\mu}_{y, 2, MC}(s)=\frac{1}{N}\sum_{k \in s}\frac{y_k}{\hat{\pi}_1^{A_\varphi}}$
  leads to a better estimation of $\mu_{y,2}$: $\hat{\mu}_{y,2,MC}(s_0)=2~649$.\\

\begin{figure}[h]
\centering
\includegraphics[width=7cm]{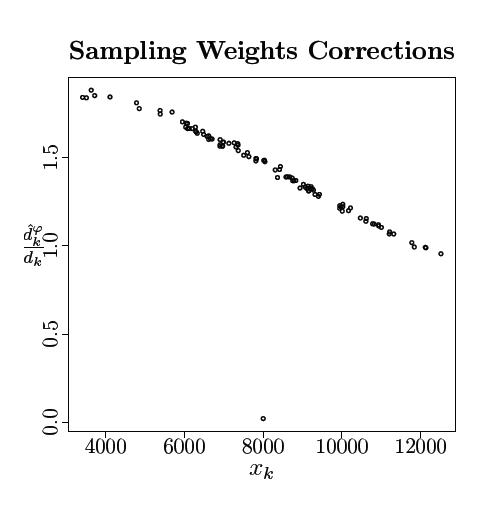}
\caption{Strata Jumper, Sampling Weight Corrections} \label{Strata
Jumper, Conditional Inclusion Probabilities}
\end{figure}

Figure \ref{Strata Jumper, Conditional Inclusion Probabilities} shows that sampling weights correction is here
a non-monotonic function of the variable $x$. We point out that the usual calibration
method would not be able to perform this kind of weights correction
because the calibration function used to correct the weights should be monotonic. \\
Similarly to the outlier setting, the
unconditional distribution of the statistics $\Phi(s)=\hat{\mu}_{x,2, HT}(s)$ has 2 modes and is far from gaussian.\\

\section{Conclusion}

At the estimation stage, a new auxiliary information can reveal that
the selected sample is imbalanced. We have shown that a conditional inference
approach can take into account this information and leads to a more precise
estimator than the unconditional Horvitz-Thompson estimator in the sense that
the conditional estimator is unbiased (conditionally and unconditionally)
and that the conditional variance is more rigorous in order to estimate the precision a posteriori.\\
In practise, we recommend to use Monte Carlo simulations in order to estimate the conditional
inclusion probabilities. \\
This technic seems particularly adapted to the treatment of outliers
and strata-jumpers.

\clearpage

\appendix

\section{Annex 1: Inclusion Probability with Conditional Poisson Sampling}
\begin{proof}
The event $\left[\sum_{l \in U, l \neq k} I_{[l \in s]}=n-1\right]$ is independent of the events $[I_{[k \in s]}=0]$ and  $[I_{[k \in s]}=1]$ in the Poisson model. So we can write:
\begin{eqnarray}
    p\left(\left[\sum_{l \in U, l \neq k} I_{[l \in s]}=n-1\right] \right)
        &=& p\left(\left[\sum_{l \in U, l \neq k} I_{[l \in s]}=n-1\right] / [I_{[k \in s]}=0] \right) \label{eqa2-1}\\
            &=& p\left(\left[\sum_{l \in U, l \neq k} I_{[l \in s]}=n-1\right] / [I_{[k \in s]}=1] \right)\label{eqa2-2}
  \end{eqnarray}
Equation (\ref{eqa2-1}) gives:
{ \footnotesize
\begin{eqnarray*}
p\left(\left[\sum_{l \in U, l \neq k} I_{[l \in s]}=n-1\right] / [I_{[k \in s]}=0] \right)
    &=&
       p\left(\left[\sum_{l \in U, l \neq k} I_{[l \in s]}=n-1\right] / [I_{[k \in s]}=0] \right)
       \\
    &=& p\left(\left[\sum_{l \in U} I_{[l \in s]}=n-1\right] / [I_{[k \in s]}=0] \right)
        \\
&=& \frac{p\left(\left[\sum_{l \in U} I_{[l \in s]}=n-1\right]  \right) p\left(  [I_{[k \in s]}=0]/\left[\sum_{l \in U} I_{[l \in s]}=n-1\right] \right)}{p\left(  [I_{[k \in s]}=0]\right)}
        \\
&=& \frac{p\left(\left[\sum_{l \in U} I_{[l \in s]}=n-1\right]  \right) (1-f_k(N,\mathbf{p},n-1))}{1-p_k},
\end{eqnarray*}
}
and equation (\ref{eqa2-2}) gives:
{ \footnotesize
\begin{eqnarray*}
p\left(\left[\sum_{l \in U, l \neq k} I_{[l \in s]}=n-1\right] / [I_{[k \in s]}=1] \right)
    &=&
       p\left(\left[\sum_{l \in U, l \neq k} I_{[l \in s]}=n-1\right] / [I_{[k \in s]}=1] \right)
       \\
    &=& p\left(\left[\sum_{l \in U} I_{[l \in s]}=n\right] / [I_{[k \in s]}=1] \right)
        \\
&=& \frac{p\left(\left[\sum_{l \in U} I_{[l \in s]}=n\right]  \right) p\left( [I_{[k \in s]}=1] / \left[\sum_{l \in U} I_{[l \in s]}=n\right]  \right)}{p\left(  [I_{[k \in s]}=1]\right)}
        \\
&=& \frac{p\left(\left[\sum_{l \in U} I_{[l \in s]}=n\right]  \right) f_k(N,\mathbf{p},n)}{p_k}.
\end{eqnarray*}
}
So we have:
\begin{eqnarray*}
f_k(U,\mathbf{p},n) &=&
                        (1-f_k(U,\mathbf{p},n-1))
                        \frac{p_k}{1-p_k}\frac{p\left(\left[\sum_{l \in U} I_{[l \in s]}=n-1\right]  \right)}{p\left(\left[\sum_{l \in U} I_{[l \in s]}=n\right]  \right)}\\
&=&
                        (1-f_k(U,\mathbf{p},n-1))
                        \frac{p_k}{1-p_k} h(U,\mathbf{p},n)
\end{eqnarray*}

And we can use the property $\sum_{k \in U} f_k(U,\mathbf{p},n)=\sum_{k \in U} \pi_k=n $ to compute
$ h(U,\mathbf{p},n)$ and conclude.
\end{proof}




\begin{thebibliography}{50}

\bibitem{Anderson(2004)}
    Anderson, P.G.
    (2004).
    A conditional perspective of weighted variance estimation of the optimal regression estimator.
    Journal of statistical planning and inference.

\bibitem{Chen et al.(1994)}
    Chen, X.-H., Dempster, A., and Liu, J.
    (1994).
    Weighted finite poplation sampling to maximize entropy.
    Biometrika, 81, 457-469.

\bibitem{Deville(2000)}
    Deville, J.C.
    (2000).
    Note sur l'algorithme de Chen, Dempster et Liu.
    Technical report, France CREST-ENSAI. [In French]

\bibitem{Fattorini (2006)}
    Fattorini L.
    (2006).
    Applying the Horvitz-Thompson criteion in complexe designs: A computer-intensive perspective for estimating inclusion probabilities.
    Biometrika, 93, 269-278.

\bibitem{Matei(2005)}
    Matei, A. and Tillé, Y.
    (2005).
    Evaluation of variance approximations and estimators in unequal probability sampling with maximum entropy.
    Journal of Official Statistics, 21, 543-570.

\bibitem{Rao(1985)}
    Rao, J.N.K.
    (1985).
    Conditional inference in survey sampling.
    Survey Methodology, 11, 15-31.

\bibitem{Robinson(1987)}
    Robinson, J.
    (1987).
    Conditioning ratio estimates under simple random sampling.
    Journal of the American Statistical Association, 82, 826-831.

\bibitem{Sarndal(1992)}
    Särndal C.-E., Swenson B. and Wretman J.
    (1992).
    Model Assisted Survey Sampling.
    Springer-Verlag, 264-269.

\bibitem{Thompson(2008)}
    Thompson M., E. and Wu C.
    (2008).
    Simulation-based Randomized Systematic PPS Sampling Under Substitution of Units.
    Survey Methodology, 34, 3-10.

\bibitem{Tillé(1998)}
    Tillé, Y.
    (1998).
    Estimation in surveys using conditional inclusion probabilities: Simple random sampling.
    International Statistical Review, 66, 303-322.
\bibitem{Tillé(1999)}
    Tillé, Y.
    (1999).
    Estimation in surveys using conditional inclusion probabilities: comlex design.
    Survey Methodology, 25, 57-66.

\end{thebibliography}
\end{document}